\documentclass[12pt,a4paper]{article}

\usepackage[utf8]{inputenc}
\usepackage{enumerate}
\usepackage{amsmath}
\usepackage{amsfonts}
\usepackage{amssymb}
\usepackage{amsthm}
\usepackage{mathtools}
\usepackage{tikz}
\usetikzlibrary{shapes.geometric,decorations.pathmorphing}
\usepackage{hyperref}
\usepackage{authblk}
\usepackage{tabularx}
\usepackage{subcaption}

\usepackage{algorithm}
\usepackage{algpseudocode}
\algnewcommand\algorithmicinput{\textbf{Input:}}
\algnewcommand\algorithmicoutput{\textbf{Output:}}
\algnewcommand\Input{\item[\algorithmicinput]}
\algnewcommand\Output{\item[\algorithmicoutput]}

\usepackage[capitalize]{cleveref}

\usepackage[backgroundcolor=gray!10]{todonotes}

\usepackage[backend=biber,isbn=false,url=false,doi=true,maxcitenames=2,maxbibnames=99]{biblatex}
\DeclareFieldFormat{title}{\mkbibquote{#1}}
\bibliography{bibliography}

\usepackage{macros}

\usepackage{fullpage}
\newtheorem{theorem}{Theorem}[section]

\newtheorem{observation}[theorem]{Observation}
\newtheorem{lemma}[theorem]{Lemma}
\newtheorem{corollary}[theorem]{Corollary}

\newtheorem{proposition}[theorem]{Proposition}

\DeclareMathOperator{\head}{head}
\DeclareMathOperator{\dist}{dist}

\DeclareMathOperator{\ind}{i}
\DeclareMathOperator{\N}{N}
\DeclareMathOperator{\lefti}{li}
\DeclareMathOperator{\righti}{ri}

\newcommand{\pvg}{terrain visibility graph}
\newcommand{\Pvg}{Terrain visibility graph}

\newcommand{\xprop}{X-property}
\newcommand{\barprop}{bar-property}

\tikzset{
	vertex/.style={circle,draw, inner sep = 1.5pt},
	edge/.style={, color=black},
	curveedge/.style={color=black,out=90,in=90},
	curvepath/.style={
		color=black,out=90,in=90,decorate,
		decoration={
			snake ,amplitude=0.25mm, segment length=1mm
		},
	},
}

\title{Advancing Through Terrains%
\footnote{A shorter version of this paper appeared in Discrete \& Computational Geometry \cite{TerrainPathJournal}.}
}

\author[1]{Vincent Froese}
\author[1]{Malte Renken\footnote{Supported by the DFG project NI~369/17-1.}}

\affil[1]{\small
  Algorithmics and Computational Complexity, Faculty~IV, TU Berlin, Berlin, Germany,\protect\\
  \{vincent.froese, m.renken\}@tu-berlin.de}

\begin{document}

\maketitle

\begin{abstract} 

  We study terrain visibility graphs, a well-known graph class closely related to polygon visibility graphs in computational geometry, for which a precise graph-theoretical characterization is still unknown. Over the last decade, terrain visibility graphs attracted attention in the context of time series analysis with various practical applications in areas such as physics, geography and medical sciences.

  We make progress in understanding terrain visibility graphs by providing several graph-theoretic
  results. For example, we show that they cannot contain antiholes of size larger than five.
  Moreover, we obtain two algorithmic results.
  We devise a fast output-sensitive shortest path algorithm on terrain-like graphs and a polynomial-time algorithm for \textsc{Dominating Set} on special terrain visibility graphs (called funnel visibility graphs).
  
\medskip

\noindent\textbf{Keywords:} computational geometry, time series visibility graphs, funnel visibility graphs, graph classes, polynomial-time algorithms
\end{abstract}

\section{Introduction}

Visibility graphs are a fundamental concept in computational geometry.
For a given set of geometrical objects (e.g.\ points, segments, rectangles, polygons) they encode which objects are visible to each other.
To this end, the objects form the vertices of the graph and there is an edge between two vertices if and only if the two corresponding objects ``can see each other'' (for a specified notion of visibility).
Visibility graphs are well-studied from a graph-theoretical perspective and 
find applications in many real-world problems occurring in different fields such as physics \cite{Elsner09,Liu10,DD12}, robotics \cite{deBerg2008}, object recognition \cite{Shapiro1979}, or medicine \cite{Ahmadlou2010,Gao16}.

In this work, we study visibility graphs of 1.5-dimensional terrains (that is, $x$-monotone polygonal chains).
This graph class has been studied since the 90's~\cite{AbelloStaircaseI} and found numerous applications in analyzing and classifying time series in recent years~\cite{Lacasa4972,Elsner09,Liu10,DD12,Stephen15,Gao16}.
However, a precise graph-theoretical characterization is an open problem.
While a necessary condition for \pvg{}s is known, it is open whether this is also sufficient (see \Cref{sec:prelims}).

In \Cref{sec:props}, we make progress towards a better understanding of terrain visibility graphs by showing that they
do not contain antiholes of size larger than five.
Moreover, we show that terrain visibility graphs do not include all unit interval graphs (which are hole-free).
Furthermore, we give an example showing that terrain visibility graphs are not unigraphs, that is, they are not uniquely determined by their degree sequence up to isomorphism.

Besides these graph-theoretical findings, our main contributions are two algorithmic results:
In \Cref{sec:shortestpath}, we develop an algorithm computing shortest paths in arbitrary induced subgraphs of terrain visibility graphs (in fact, the algorithm even works for a more general class of graphs known as \emph{terrain-like} graphs~\cite{AFKS19} satisfying a weaker condition) in~$\bigO(d^*\log \Delta)$ time, where~$d^*$ is the length of the shortest path and $\Delta$ is the maximum degree (also~$\bigO(d^*)$ time is possible with an $\bigO(n^2)$-time preprocessing).
\Cref{sec:dominatingset} presents an~$\bigO(n^4)$-time algorithm for \textsc{Dominating Set} on a known subclass of terrain visibility graphs called funnel visibility graphs.

\paragraph*{Related Work.} For a general overview on visibility graphs and related problems see the survey by \textcite{GG13}.
As regards the origin of terrain visibility graphs, \textcite{AbelloStaircaseI} studied visibility graphs of staircase polygons which are closely related to terrain visibility graphs as \textcite{Colley92} showed that they are in one-to-one correspondence with the core induced subgraphs of staircase polygon visibility graphs.
\textcite{AbelloStaircaseI} described three necessary properties that are satisfied by every terrain visibility graph.
Recently, \textcite{AGKSW20} showed that these properties are not sufficient.
\textcite{Saeedi2015} gave a simpler proof for the results of \textcite{AbelloStaircaseI}.
\textcite{AbelloStaircaseUniform} showed that visibility graphs of staircase polygons with unit step-length can be recognized via linear programming and that not all staircase polygon visibility graphs can be represented with unit steps.
\textcite{ChoiFunnels} studied another subclass of terrain visibility graphs called \emph{funnel visibility graphs} which is linear-time recognizable (also studied by~\textcite{ColleyTowers97}).

\textcite{Lacasa4972} introduced terrain visibility graph (under the name of \emph{time series visibility graphs}) in the context of time series analysis.
A variant (called \emph{horizontal visibility graphs}) where two vertices can only see each other horizontally was later introduced by \textcite{LLBL09}.
Horizontal visibility graphs were fully characterized by~\textcite{GUTIN20112421} who showed that these are exactly the outerplanar graphs with a Hamiltonian path.
Moreover, \textcite{Luque2017} showed that certain canonical horizontal visibility graphs are uniquely determined by their degree sequence.

Notably, the \textsc{Terrain Guarding} problem, that is, selecting~$k$ terrain points that guard the whole terrain (which is closely related to \textsc{Dominating Set} on terrain visibility graphs) has been extensively studied in the literature and is known to be NP-hard~\cite{KingKrohnHardness} even on \emph{orthogonal terrains} \cite{BG18}.
It has recently been studied from a parameterized perspective \cite{AFKSZ18} and also from an approximation point of view~\cite{AFKS19}.

\section{Preliminaries}\label{sec:prelims}
We assume the reader to be familiar with basic concepts and classes of graphs (refer to \textcite{Brandstadt99} for an overview).

A (1.5-dimensional) \emph{terrain} is an $x$-monotone polygonal chain in the plane defined by a set $V$ of \emph{terrain vertices} with pairwise different~$x$-coordinates.
For two terrain vertices $p, q$, we write $p < q$ if $p$ is ``left'' of~$q$, that is, $p^x < q^x$, where $p^x$ denotes the~$x$-coordinate of~$p$.
Furthermore, we define $[p, q] := \{x\mid p \leq x \leq q\}$.
The corresponding \emph{terrain visibility graph} is defined on the set of terrain vertices where two vertices~$p$ and~$q$ are adjacent if and only if they \emph{see each other}, that is, there is no vertex between them that lies on or above the line segment connecting $p$ and~$q$ (see \Cref{fig:example_tvg} for an example).
Formally, there exists an edge~$\{p,q\}$, for $p < q$, if and only if all terrain vertices $r$ with $p < r < q$ satisfy
\begin{align*}
r^y < p^y + (q^y - p^y) \frac{r^x - p^x}{q^x - p^x}\,.
\end{align*}

\begin{figure}[t]
  \centering
  \begin{tikzpicture}
    \tikzset{every node/.style={vertex}}
    \node (1) at (0,2) {};
    \node (2) at (1,0) {};
    \node (3) at (2,1) {};
    \node (4) at (2.5,-1) {};
    \node (5) at (3.5,4) {};
    \node (6) at (5,3) {};
    \draw (1) -- (2) -- (3) -- (4) -- (5) -- (6);
    \draw (1) -- (3);
    \draw (1) -- (5);
    \draw (2) -- (5);
    \draw (3) -- (5);
  \end{tikzpicture}\hspace{2cm}
  \begin{tikzpicture}[baseline=-2cm]
    \tikzset{every node/.style={vertex}}
    \node (1) at (1,0) {};
    \node (2) at (2,0) {};
    \node (3) at (3,0) {};
    \node (4) at (4,0) {};
    \node (5) at (5,0) {};
    \node (6) at (6,0) {};
    \draw (1) -- (2) -- (3) -- (4) -- (5) -- (6);
    \draw[curveedge] (1) to (3);
    \draw[curveedge] (1) to (5);
    \draw[curveedge] (2) to (5);
    \draw[curveedge] (3) to (5);
  \end{tikzpicture}
  \caption{A terrain visibility graph drawn in two different ways (with a corresponding terrain on the left).}
  \label{fig:example_tvg}
\end{figure}
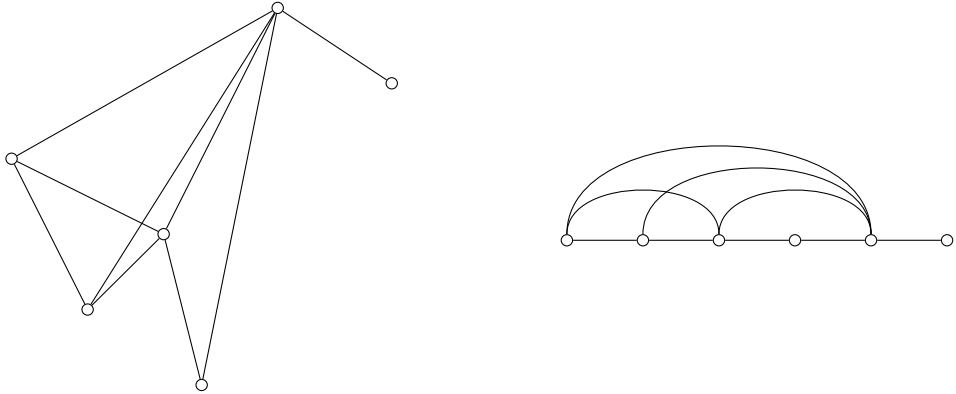

\textcite{Hershberger87} gave an algorithm to compute the visibility graph of a given terrain with a running time that is linear in the size of the graph.

Let $p \in V$ be a terrain vertex and $q, r$ be the vertices immediately to its left and right.
We call $p$ \emph{convex} if $q$ and $r$ see each other.
Otherwise it is called \emph{reflex}.
The leftmost and rightmost vertex of a terrain are neither convex nor reflex.

Clearly, every terrain visibility graph contains a Hamiltonian path along the order of the terrain vertices.
Moreover, the visibility graph of a terrain is invariant under some affine transformation, in particular translation, scaling,
or vertical shearing (i.e., a map $(x, y) \mapsto (x, mx+y)$ with $m \in \RR$).

The following are two elementary properties of \pvg{s} (we stick to the names coined by \textcite{Saeedi2015}).
Two edges $\{p, q\}$ and $\{r, s\}$ are said to be \emph{crossing} if the corresponding line segments cross, i.e.\ if $p < r < q < s$.
The first property states that two crossing edges imply the existence of another edge.

\begin{center}
\begin{minipage}{0.95\textwidth}
\textbf{\xprop{}:}
Let $p, q, r, s$ be four terrain vertices with $p < q < r < s$.
If $p$ sees $r$ and $q$ sees $s$, then $p$ sees $s$.
\end{minipage}
\end{center}

The X-property holds because $q$ must lie below the line segment through $p, r$ while $r$ must lie below the line segment through $q, s$.
Thus, any point on the line segment between $p$ and $s$ is above one (or both) of these two line segments, so it cannot be obstructed by any vertex.

\begin{center}
  \begin{minipage}{0.95\textwidth}
\textbf{\barprop{}:}
If $p, q$ are two non-consecutive terrain vertices that can see each other,
then there is a vertex $r$ between $p$ and $q$ that can see both of them.
\end{minipage}
\end{center}

The bar-property is immediate if we first apply a vertical shear mapping such that $p^y = q^y$.
Then, $r$ is simply the vertex between $p$ and $q$ which has maximal $y$-coordinate.

Any vertex-ordered graphs with a Hamiltonian path following the order which satisfies the \xprop{} and \barprop{} is called \emph{persistent}.
\textcite{AGKSW20} constructed a persistent graph which is not a terrain visibility graph, thus proving that persistent graphs form a proper superset of \pvg{}s.

\section{Graph Properties}\label{sec:props}

In this section we prove some graph-theoretical results.
These mostly apply not only to \pvg{}s but to persistent graphs.

\subsection{Induced Subgraphs}

It is known that terrain visibility graphs can contain an induced~$C_4$ \cite{Everett90,Ghosh97}.
It is further known that, whenever four vertices $p_1, p_2, p_3, p_4$ of a persistent graph form an induced~$C_4$ (in that order), then they cannot satisfy $p_1 < p_2 < p_3 < p_4$ since this violates either the \xprop{} or the \barprop{}~\cite{Ghosh97,Abello04}. This leads to the following simple observation.

\begin{observation}\label{thm:C4}
	If four vertices $p_1, p_2, p_3, p_4$ of a persistent graph form an induced $C_4$ (in that order),
	then the two leftmost of these are either $p_1$ and $p_3$ or $p_2$ and $p_4$.
\end{observation}

\begin{proof}
  Assume that the leftmost vertex of~$p_1,p_2,p_3,p_4$ has an edge to the second leftmost vertex.
  If it also has an edge to the rightmost vertex, then this case is equivalent (up to cyclic shifting and order reversal) to the case~$p_1 < p_2 < p_3 < p_4$, which is not possible~\cite{Ghosh97,Abello04}.
  Hence, it has an edge to the second rightmost vertex.
  But then the \xprop{} is violated since there is also an edge from the second leftmost to the rightmost vertex.

  Thus, there cannot be an edge between the two leftmost vertices, which means these are either~$p_1$ and~$p_3$ or~$p_2$ and~$p_4$.
\end{proof}

We can generalize this observation to larger induced cycles: While $C_k$ may appear as induced subgraph for any $k$, its vertices can only occur in a specific order.
For an example consider \cref{fig:C6}. Note that this construction can be generalized to any $k \geq 4$ by changing the number of vertices on the bottom middle path.

\begin{figure}
  \centering
\begin{tikzpicture}[yscale=0.2,xscale=0.3]
	\tikzset{every node/.style={vertex,inner sep=1pt}}
	\node (2) at (0,37) {};
	\node (x1) at (5,32) {};
	\node (6) at (6,12) {};
	\node (5) at (16,14) {};
	\node (4) at (20,14) {};
	\node (3) at (30,12) {}; 
	\node (x2) at (31,32) {};
	\node (1) at (36,37) {};
	
	\foreach \x in {1,...,6}
		\draw (x1) -- (\x) -- (x2);
	\draw (x1) -- (x2);
	\draw[line width=1.5pt] (1) -- (2) -- (3) -- (4) -- (5) -- (6) -- (1);
\end{tikzpicture}
\caption{A \pvg{} containing an induced $C_6$ (bold edges).}
\label{fig:C6}
\end{figure}

We start with the following basic lemma.

\begin{lemma}\label{thm:obstacle_lemma}
Let $G$ be a persistent graph and $p, q, r$ three of its vertices with $p < q < r$.
If $G$ contains edges $\{p,q\}$ and $\{p, r\}$ but not $\{q, r\}$,
then $p$~and~$q$ have a common neighbor $b \in V(G)$ with $q < b < r$.
\end{lemma}
\begin{proof}
Pick $b$ as the leftmost neighbor of $p$ which satisfies $q < b \leq r$.
Pick further $c$ as the rightmost neighbor of $b$ which satisfies $p \leq c \leq q$.
If $c = q$ then we are done, so assume for contradiction $c \neq q$.
By the \barprop{}, $c$~and~$b$ must have a common neighbor $d$ with $c < d < b$, and by choice of $c$ we have $q < d$.
If $p \neq c$, then we can apply the \xprop{} to $p < c < q < d$ to obtain $\{p, d\} \in E(G)$.
(If $p = c$ then this is obvious.)
But this means that $d$ contradicts the fact that $b$ was chosen leftmost.
\end{proof}

For the vertex ordering of induced cycles, we can now derive the following.

\begin{lemma}\label{thm:hole_lemma}
Let $G$ be a persistent graph and let $p_1, \dots, p_k \in V(G)$ be any $k$ vertices that form an induced cycle (in this order) with $k \geq 4$.
If $p_2$ is the leftmost of these, then $p_4, \dots, p_k$ must all be to the left of both, $p_1$ and $p_3$.
\end{lemma}
\begin{proof}
Assume without loss of generality that $p_1 > p_3$.
Since $k \geq 4$, by \cref{thm:obstacle_lemma}, there is a vertex $q$ between $p_3$ and $p_1$ which is a common neighbor of $p_2$ and $p_3$.

By the \xprop{}, $p_4, \dots, p_{k-1}$ must all lie between $p_2$ and $p_1$.
Suppose that there exists $i > 3$ (chosen minimally) such that $p_{i} > q$.
Then, by the \xprop{}, $\{p_2, p_i\} \in E(G)$, which is a contradiction.
Since $p_3$ has an edge to $q$ but not to~$p_1$, $p_k$ must even lie to the left of $p_3$ by the \xprop{}.

Now, assume that there exists $j > 3$ (chosen maximally) such that $p_j > p_3$.
Then, $\{p_j, p_2\} \in E(G)$ by the \xprop{}.
Therefore, $p_4, \ldots, p_k$ must all be to the left of $p_3$.
\end{proof}

\Cref{thm:hole_lemma} now leads to a specific vertex ordering for an induced cycle.

\begin{proposition}\label{thm:Ck}
	Let $G$ be a persistent graph and let $p_1, \dots, p_k \in V(G)$ form an induced cycle (in this order) with $k \geq 4$.
	Then, $p_1 > p_3 > p_4 > \dots > p_k > p_2$ holds up to cyclic renaming and order reversal.
\end{proposition}
\begin{proof}
	We assume without loss of generality that $p_2$ is the leftmost vertex and that $p_1 > p_3$.
	By \cref{thm:hole_lemma}, it follows $p_4, \dots, p_k \subset [p_2, p_3]$.
	In particular, $p_1$ is the rightmost vertex.
	Thus, a mirrored version of \cref{thm:hole_lemma} applies to the shifted sequence $p_k, p_1, p_2, \dots, p_{k-1}$,
	giving $p_3, p_4, \dots, p_{k-1} \subset [p_k, p_1]$.
	Overall, we then have $p_4, p_5, \dots, p_{k-1} \subset [p_k, p_3]$.

        Now, assume towards a contradiction that there exists $3 < i < k-1$ with $p_i < p_{i+1}$, where $i$ is chosen minimally.
	Let $j$ be minimal such that $p_j < p_i$, that is, $p_{j-1} > p_i$. Note that $j > i$ by the choice of~$i$.
        Let $\ell < i$ be maximal with $p_\ell > p_{j-1}$, that is, $p_{\ell+1} < p_{j-1}$.
        Note also that~$p_{\ell+1}\ge p_i$ by the choice of~$i$.
	Then, $p_\ell > p_{j-1} > p_{\ell+1} > p_j$. Hence, by the \xprop{}, $G$ contains the edge $\{p_\ell, p_j\}$.
	Since $p_3, p_4, \dots, p_k$ form an induced path, this implies $j = \ell+1$, contradicting the fact that $j > i > \ell$.	
\end{proof}

Interestingly, we can use \Cref{thm:C4} to show that persistent graphs do not contain large antiholes (induced subgraphs isomorphic to the complement of a cycle).

\begin{theorem}\label{thm:antiholes}
\Pvg{s} do not contain antiholes of size at least~6.
\end{theorem}
\begin{proof}
Let $p_1, \dots, p_k$ induce an antihole in $G$ in this order with $k \geq 6$.
Observe that for any $i, j$ with $\abs{i-j} \notin \{0, 1, 2\} \pmod{k}$, the vertices $p_i$, $p_j$, $p_{i+1}$, $p_{j+1}$ form an induced 4-cycle in this order.

Assume without loss of generality that $p_1$ is the rightmost vertex.
Then, for each $j =4,\ldots,k-2$, we apply \cref{thm:C4} to the $C_4$ on $p_1, p_j, p_2, p_{j+1}$ which yields that either $p_j$ and~$p_{j+1}$ or~$p_1$ and~$p_2$ are the two rightmost.
It follows by assumption that~$p_1$ and~$p_2$ are the rightmost. Hence,
$p_{5},\ldots,p_{k-1}$ are all to the left of~$p_{2}$.
Now, for $j =5,\ldots,k-1$, we apply \cref{thm:C4} to $p_2, p_j, p_3, p_{j+1}$ and obtain that $p_5, \dots, p_k$ are also to the left of $p_3$.
Finally, for $j = 6, \dots, k$ we use the same argument on $p_3, p_j, p_4, p_{j+1}$ (where~$p_{k+1}=p_1$) and obtain that~$p_6, \dots, p_k, p_1$ are to the left of $p_3$.
This contradicts our assumption that~$p_1$ is the rightmost vertex.
\end{proof}

Considering possible subclasses of persistent graphs,
we close this subsection by showing that not even connected unit interval graphs are necessarily persistent.
A demonstration of this fact is the unit interval graph depicted in \Cref{fig:unit_example}.
It is not persistent because it only has one Hamiltonian path (up to isomorphism)
and the ordering given by this path violates the \xprop{}.
It is open whether every unit interval graph can appear as an induced subgraph of a persistent graph.

\begin{figure}
  \centering
\begin{tikzpicture}
  \tikzset{every node/.style={vertex,inner sep=2pt}}
  \node (1) at (0,0) {};
  \node (2) at (1,0) {};
  \node (3) at (2,-1) {};
  \node (4) at (2,1) {};
  \node (5) at (3,0) {};
  \node (6) at (4,0) {}; 
  \draw (1) -- (2) -- (3) -- (5);
  \draw (2) -- (4) -- (3);
  \draw (4) -- (5) -- (6);
\end{tikzpicture}
\hspace{2em}
\begin{tikzpicture}
  \draw[thick] (-0.25,0) -- (0.75,0);
  \draw[thick] (0.5,0.5) -- (1.5,0.5);
  \draw[thick] (1.25,1.0) -- (2.25,1.0);
  \draw[thick] (1.25,0) -- (2.25,0);
  \draw[thick] (2,0.5) -- (3,0.5);
  \draw[thick] (2.75,0) -- (3.75,0);
\end{tikzpicture}
\caption{A unit interval graph (with interval representation on the right) which is not a terrain visibility graph.}
\label{fig:unit_example}
\end{figure}
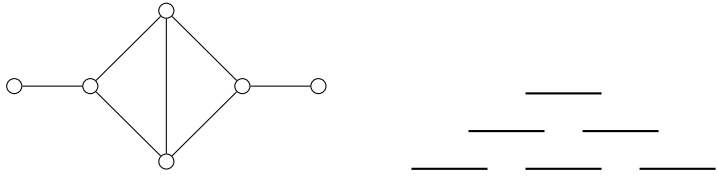

\subsection{Degree Sequences}

\textcite{Luque2017} studied the degree sequences of horizontal visibility graphs
in order to explain why measures based on the degree sequence of horizontal visibility graphs of time series perform well in classification tasks. Their conclusion is that the degree sequence essentially contains all information of the underlying time series. Formally, they show that (canonical) horizontal visibility graphs are uniquely determined by their degree sequence.

In contrast, this not the case for terrain visibility graphs since
the two terrain visibility graphs in \Cref{fig:G1,fig:G2} both have the ordered degree sequence $(7,4,3,4,5,7,4,4,4,6,4)$ and are not isomorphic (since the unique degree-3 vertex has a degree-7 neighbor in one graph but not in the other).

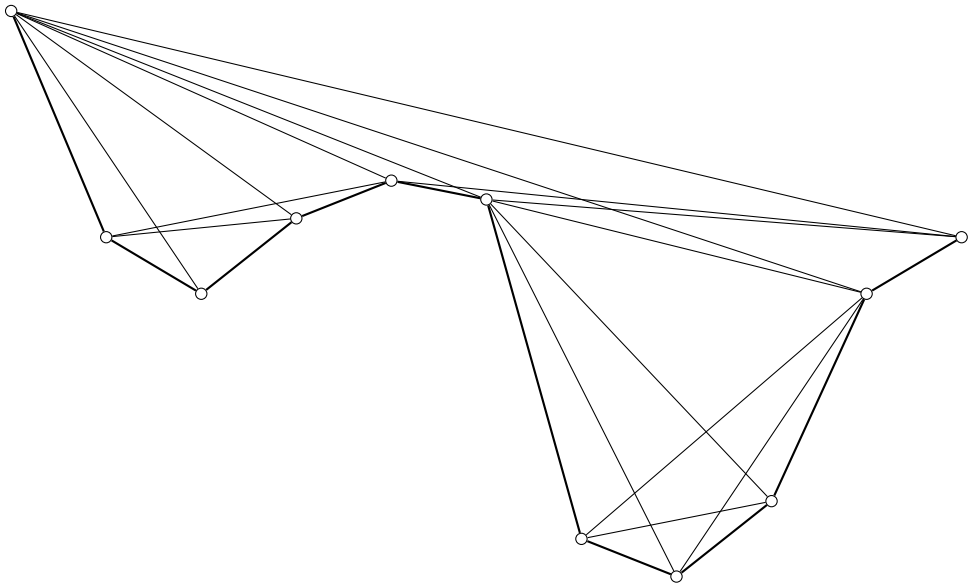
\begin{figure}
  \centering
  \begin{tikzpicture}[yscale=0.25,xscale=1.25]
    \tikzset{every node/.style={vertex}}
    \node (0) at (0,30) {};
    \node (1) at (1,18) {};
    \node (2) at (2,15) {};
    \node (3) at (3,19) {};
    \node (4) at (4,21) {};
    \node (5) at (5,20) {};
    \node (6) at (6,2) {};
    \node (7) at (7,0) {};
    \node (8) at (8,4) {};
    \node (9) at (9,15) {};
    \node (10) at (10,18) {};
    \draw[thick] (0) -- (1) -- (2) -- (3)
    -- (4) -- (5) -- (6) -- (7) -- (8) -- (9) -- (10);
    \draw (0) -- (2);
    \draw (0) -- (3);
    \draw (0) -- (4);
    \draw (0) -- (5);
    \draw (0) -- (9);
    \draw (0) -- (10);
    \draw (1) -- (3);
    \draw (1) -- (4);
    \draw (4) -- (10);
    \draw (5) -- (7);
    \draw (5) -- (8);
    \draw (5) -- (9);
    \draw (5) -- (10);
    \draw (6) -- (8);
    \draw (6) -- (9);
    \draw (7) -- (9);
  \end{tikzpicture}
  \caption{A terrain visibility graph~$G_1$ with ordered degree sequence $(7,4,3,4,5,7,4,4,4,6,4)$.
  Vertices have unit-spaced $x$-coordinates. The corresponding $y$-coordinates are $30,18,15,19,21,20,2,0,4,15,18$. Note that in the drawing the~$y$-axis is scaled down.}
  \label{fig:G1}
\end{figure}

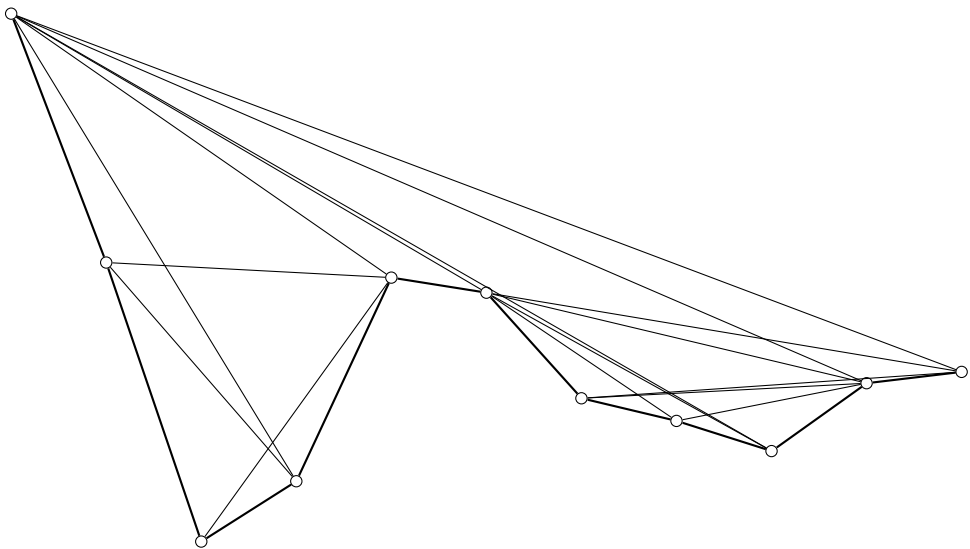
\begin{figure}
  \centering
  \begin{tikzpicture}[yscale=0.05,xscale=1.25]
    \tikzset{every node/.style={vertex}}
    \node (0) at (0,140) {};
    \node (1) at (1,74) {};
    \node (2) at (2,0) {};
    \node (3) at (3,16) {};
    \node (4) at (4,70) {};
    \node (5) at (5,66) {};
    \node (6) at (6,38) {};
    \node (7) at (7,32) {};
    \node (8) at (8,24) {};
    \node (9) at (9,42) {};
    \node (10) at (10,45) {};
    \draw[thick] (0) -- (1) -- (2) -- (3)
    -- (4) -- (5) -- (6) -- (7) -- (8) -- (9) -- (10);
    \draw (0) -- (3);
    \draw (0) -- (4);
    \draw (0) -- (5);
    \draw (0) -- (8);
    \draw (0) -- (9);
    \draw (0) -- (10);
    \draw (1) -- (3);
    \draw (1) -- (4);
    \draw (2) -- (4);
    \draw (5) -- (7);
    \draw (5) -- (8);
    \draw (5) -- (9);
    \draw (5) -- (10);
    \draw (6) -- (9);
    \draw (6) -- (10);
    \draw (7) -- (9);
  \end{tikzpicture}
  \caption{A terrain visibility graph~$G_2$ with ordered degree sequence $(7,4,3,4,5,7,4,4,4,6,4)$.
    Vertices have unit-spaced $x$-coordinates. The corresponding $y$-coordinates are $140,74,0,16,70,66,38,32,24,42,45$. Note that in the drawing the~$y$-axis is scaled down.}
  \label{fig:G2}
\end{figure}

\section{Shortest Paths}\label{sec:shortestpath}

\newcommand{\lhorizon}[1]{\operatorname{lhorizon}(#1)}
\newcommand{\rhorizon}[1]{\operatorname{rhorizon}(#1)}

\newcommand{\lclosest}[2]{\operatorname{lclosest}_{#1}(#2)}
\newcommand{\rclosest}[2]{\operatorname{rclosest}_{#1}(#2)}

\newcommand{\overiota}[2]{%
\hbox{\scalebox{0.8}{$#1\iota$}}
\raisebox{\raiseamount{#1}}{\scalebox{0.5}{$#1\mkern-11mu#2$}}
}

\def\raiseamount#1{%
  \ifx#1\displaystyle
    .9ex
  \else
    \ifx#1\textstyle
      .9ex
    \else
      \ifx#1\scriptstyle
        .6ex
      \else
        .45ex
      \fi
    \fi
  \fi
}

A natural example for the occurrence of \pvg{s} is a network of stations communicating via line-of-sight links, e.g.\ radio signals.
A common task is to determine the shortest path between two vertices $s < t$.
If the length of a path is measured via Euclidean distance, then the easy solution is to always go to the right as far as possible without going beyond~$t$.
In general, computing Euclidean shortest paths in polygon visibility graphs is a well-studied problem and linear-time algorithms are known~\cite{GHLST87}.

But a more realistic distance measure in the above scenario is the number of edges, as edge travel times might be negligible in comparison to the processing times at the vertices.
In this setting, the situation becomes more difficult since
it might now be better to move in the opposite direction first.
Nevertheless, the ``go as far as possible'' principle still proves very useful here.
To the best of our knowledge, no specific algorithm for unweighted shortest path computation in terrain (or polygon) visibility graphs has been developed so far.

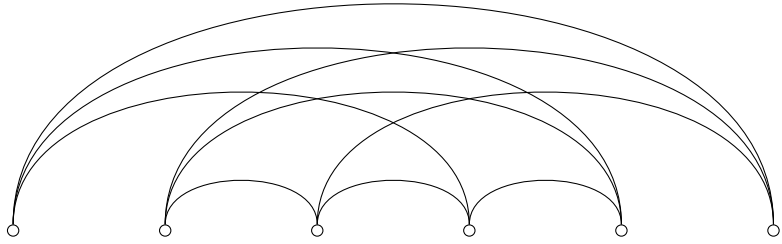
\begin{figure}
  \centering
  \begin{tikzpicture}[scale=2.0]
    \tikzset{every node/.style={vertex}}
    \node (1) at (0,0) {};
    \node (6) at (5,0) {};
    \node (2) at (1,0) {};
    \node (3) at (2,0) {};
    \node (4) at (3,0) {};
    \node (5) at (4,0) {};
    \draw[curveedge] (1) to (6);
    \draw[curveedge] (1) to (4);
    \draw[curveedge] (1) to (5);
    \draw[curveedge] (6) to (2);
    \draw[curveedge] (6) to (3);
    \draw[curveedge] (2) to (3);
    \draw[curveedge] (4) to (5);
    \draw[curveedge] (3) to (4);
    \draw[curveedge] (2) to (5);
  \end{tikzpicture}
  \caption{An ordering of the vertices of an antihole of size 6 that satisfies the \xprop{}.}
  \label{fig:antihole}
\end{figure}

The algorithm we describe in this section does not only work for \pvg{s} but
in fact for every graph with a vertex ordering fulfilling the \xprop{} (so-called \emph{terrain-like} graphs~\cite{AFKS19}).
Note that every induced subgraph of a \pvg{} is terrain-like.
Interestingly, the converse is not true since the complement of~$C_6$ (that is, a size-6 antihole)
can satisfy the \xprop{} (see \Cref{fig:antihole}) but cannot be an induced subgraph of a \pvg{} (\Cref{thm:antiholes}).
Hence, our algorithm can be used in a more general context.
For example, in the communication scenario above, we can also handle vertices which obstruct communication but are not stations themselves.

In this section, we assume $G$ to have a (known) vertex ordering satisfying the \xprop{}.
Furthermore, we assume $s$ and $t$ to be two vertices of $G$ with $s < t$ and $\dist(s, t) =d^* < \infty$,
where $\dist(s, t)$ denotes the length (that is, the number of edges) of a shortest path from $s$ to $t$.

We start with the crucial observation that a shortest path contains at most one pair of crossing edges.

\begin{lemma}\label{thm:shortest_path_crossing}
	If a shortest $s$-$t$-path $P$ contains a pair of crossing edges $\{v, v'\}, \{w, w'\}$ with $v < w < v' < w'$, 
	then it also contains the edge $\{v, w'\}$.
	Furthermore, this is the only pair of crossing edges in~$P$ and $V(P) \subseteq [v, w']$.
      \end{lemma}

      \begin{figure}
  \centering
  \begin{tikzpicture}
    \node[vertex,label=below:$\strut{}x$] (x) at (0,0) {};
    \node[vertex,label=below:$\strut{}v$] (v) at (1,0) {};
    \node[vertex,label=below:$\strut{}x'$] (x') at (2,0) {};
    \node[vertex,label=below:$\strut{}w$] (w) at (3,0) {};
    \node[vertex,label=below:$\strut{}v'$] (v') at (4,0) {};
    \node[vertex,label=below:$\strut{}w'$] (w') at (5,0) {};
    \draw[curveedge,dashed] (x) to (x');
    \draw[curveedge] (v) to (v');
    \draw[curveedge] (w') to (w);
    \draw[curveedge] (v) to (w');
    \draw[curveedge,dashed] (x) to (w');
  \end{tikzpicture}\hspace{2cm}
  \begin{tikzpicture}
    \node[vertex,label=below:$\strut{}v$] (v) at (0,0) {};
    \node[vertex,label=below:$\strut{}w$] (w) at (1,0) {};
    \node[vertex,label=below:$\strut{}x$] (x) at (2,0) {};
    \node[vertex,label=below:$\strut{}v'$] (v') at (3,0) {};
    \node[vertex,label=below:$\strut{}x'$] (x') at (4,0) {};
    \node[vertex,label=below:$\strut{}w'$] (w') at (5,0) {};
    \draw[curveedge,dashed] (x) to (x');
    \draw[curveedge] (v) to (v');
    \draw[curveedge] (w') to (w);
    \draw[curveedge] (v) to (w');
    \draw[curveedge,dashed] (v) to (x');
  \end{tikzpicture}
  \caption{Two sketches of a shortest path containing the edges~$\{v,v'\}$, $\{v,w'\}$ and $\{w,w'\}$.
  The dashed edge~$\{x,x'\}$ cannot also be on the path since this would imply the existence of the other dashed edge.}
\label{fig:crossing}
\end{figure}
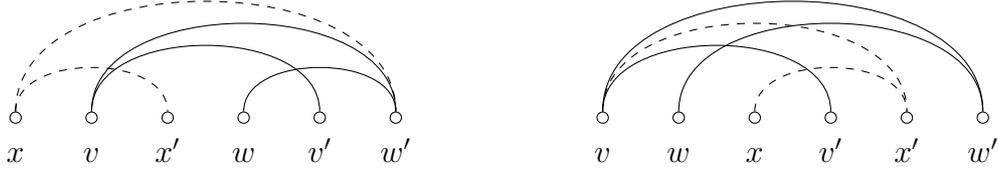
      
\begin{proof}
	Let $P$ be a shortest $s$-$t$-path and $\{v, v'\}, \{w, w'\}$ two crossing edges with $v < w < v' < w'$.
	By the \xprop{}, $G$ contains the edge $\{v, w'\}$. Since~$P$ is a shortest path and thus also an induced path, it must contain that edge (see \Cref{fig:crossing}).
	We claim that no other edge $\{x, x'\}$, $x < x'$, of~$P$ can cross $\{v, v'\}$.
	If $x < v < x' < v'$, then~$\{x,x'\}$ would also cross $\{v, w'\}$.
        Hence, by the \xprop{}, $P$ would have to contain the edge~$\{x,w'\}$, which is not possible since~$w'$ would be incident to three edges of~$P$.
	Otherwise, if $v < x < v'< x'$, then $P$ would contain the edge $\{v, x'\}$, which is again not possible.
	By symmetry, no other edge of $P$ can cross $\{w, w'\}$.
        
	Consequently, $P$ contains only vertices in~$[v, w']$ since otherwise~$P$ would contain an edge from a vertex~$u\not\in[v,w']$ to a vertex~$x$ with~$v \leq x \leq w'$.
	Since $v, w'$ are both already contained in two edges of $P$, this would imply $v < x < w'$ and thus the edge~$\{u,x\}$ would cross~$\{v,v'\}$ or~$\{w,w'\}$.
	
	It follows that $P$ cannot contain a second pair of crossing edges since the same argument would apply to that pair.
\end{proof}

We denote a shortest $s$-$t$-path $P$ of length~$d^*$ by its vertices $s=p_0, p_1, \ldots, p_{d^*}=t$
and define $\lefti(P)$ to be the index of the leftmost vertex of $P$.
Analogously, $\righti(P)$ is the index of the rightmost vertex of $P$.
The following can be obtained from \Cref{thm:shortest_path_crossing} (see also \cref{fig:path_cases} (iii)).

\begin{lemma}\label{thm:right_left}
	Let $P$ be a shortest $s$-$t$-path.
	Then, $\righti(P) < \lefti(P)$ holds if and only if $P$ contains a pair of crossing edges.
        
        Moreover, if~$P$ contains a pair of crossing edges, then $\righti(P) = \lefti(P)-1$ and
        $p_i < p_j$ holds for all $i < \righti(P) < \lefti(P) < j$.
\end{lemma}
\begin{proof}
  If $\righti(P) < \lefti(P)$, then $P$ must contain a pair of crossing edges between the subpaths from $s$ to $p_{\righti(P)}$ and from $p_{\lefti(P)}$ to $t$.
  Conversely, if $P$ contains a pair of crossing edges, then, by \cref{thm:shortest_path_crossing}, the two crossing edges must be $\{p_{\righti(P)-1}, p_{\righti(P)}\}$ and $\{p_{\lefti(P)}, p_{\lefti(P)+1}\}$.
        This implies that~$P$ contains the edge~$\{p_{\righti(P)},p_{\lefti(P)}\}$, which implies $\righti(P) = \lefti(P)-1$.
	Then, for all $i < \righti(P) < \lefti(P) < j$, it holds that $p_i < p_j$ since $P$ has no other crossing edges and $p_{\righti(P)-1} < p_{\lefti(P) + 1}$.
\end{proof}

Clearly, \Cref{thm:right_left} implies that $\lefti(P) < \righti(P)$ if and only if $P$ contains no crossing edges.
Moreover, the following holds (see \cref{fig:path_cases} (i) and (ii)).

\begin{lemma}\label{thm:left_right}
	Let $P$ be a shortest $s$-$t$-path with $\lefti(P) < \righti(P)$.
	Then, $p_i < p_j$ holds for all $i < \lefti(P) < j$ and all $i < \righti(P) < j$.
\end{lemma}
\begin{proof}
We prove both cases simultaneously.
Assume towards a contradiction that $p_j < p_i$,
then we have~$p_{\lefti(P)} < p_j < p_i < p_{\righti(P)}$.
Consider the subpath $P_1$ of $P$ that connects $p_i$ with $p_{\lefti(P)}$ and the subpath $P_2$ of $P$ that connects $p_j$ with $p_{\righti(P)}$.
These two subpaths are vertex-disjoint and must thus be crossing, meaning that there is a pair of crossing edges $e_1, e_2$ where $e_1$ is an edge of $P_1$ and $e_2$ an edge of $P_2$.
By \cref{thm:right_left}, this implies $\righti(P) < \lefti(P)$, which is a contradiction.
\end{proof}

In combination, \cref{thm:right_left} and \cref{thm:left_right} imply the following corollary.
\begin{corollary}\label{cor:vertex-st-order}
  For a shortest $s$-$t$-path~$P$, it holds that all vertices $p_i$ with $i < \righti(P)$ satisfy $p_i < t$ and all vertices $p_i$ with $i > \lefti(P)$ satisfy $p_i > s$.
\end{corollary}

The results derived above characterize the global structure of a shortest $s$-$t$-path.
Next, we will investigate the local structure.
The most important consequence of the \xprop{} is that large steps are usually better than small steps.
We formalize this in the following.

\begin{figure}
	\centering
	\begin{tikzpicture}[xscale=1.3, every label/.style={align=center}]
		\node[vertex,label=below:$\strut{}\alpha_s(2)$] (alpha_s2) at (-2,0) {};
		\node[vertex,label=below:$\strut{}\alpha_s(1)$] (alpha_s1) at (-1,0) {};
		\node[vertex,label=below:$\strut{}s$\\$\strut{}\beta_t(4)$] (s) at (0,0) {};
		\node[vertex,label=below:$\strut{}\beta_s(1)$\\$\strut{}\beta_t(3)$] (beta_s1) at (1,0) {};
		\node[vertex,label=below:$\strut{}\beta_s(2)$\\$\strut{}\beta_t(2)$] (beta_s2) at (2,0) {};
		\node[vertex,label=below:$\strut{}\beta_t(1)$] (beta_t1) at (3,0) {};
		\node[vertex,label=below:$\strut{}t$] (t) at (4,0) {};
		\node[vertex] (w) at (5,0) {};
		\node[vertex,label=below:$\strut{}\alpha_t(1)$] (alpha_t1) at (6,0) {};
		\draw[curveedge]
			(alpha_s2) to (alpha_s1) to (s) to (beta_s1) to (beta_s2)
			(beta_t1) to (t) to (w) to (alpha_t1)
			(alpha_s1) to (alpha_t1)
			(beta_s2) to (alpha_t1)
			(t) to (alpha_t1)
			;
	\end{tikzpicture}
	\caption{
		Example graph with all $\alpha$- and $\beta$-vertices labeled.
		Note that $\alpha_s(2) = \alpha_s(3) = \ldots$, $\beta_s(2)=\beta_s(3)=\ldots$, $\alpha_t(1)=\alpha_t(2)=\ldots$, and $\beta_t(4)=\beta_t(5)=\ldots$.
		Furthermore, $s = \alpha_s(0) = \beta_s(0)$ and $t = \alpha_t(0) = \beta_t(0)$.
	}
	\label{fig:alpha_beta}
\end{figure}
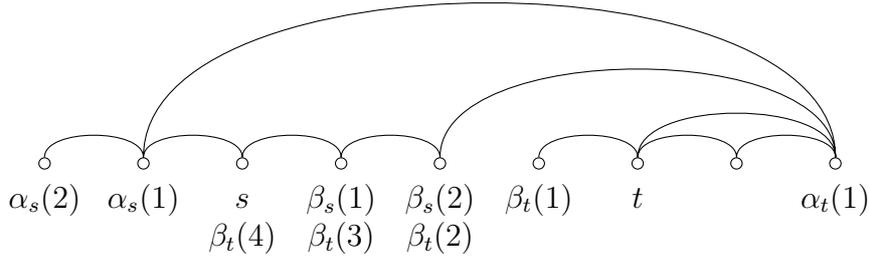

Define $\alpha_s(k)$ as the \textit{leftmost} and $\beta_s(k)$ as the \textit{rightmost} vertex that can be reached from $s$ by a path of length at most $k$
that only uses vertices $v \leq t$.
We symmetrically define $\alpha_t(k)$ as the \textit{rightmost} and $\beta_t(k)$ as the \textit{leftmost} vertex reachable from $t$ by a path of length at most $k$
using only vertices $v \geq s$.
Clearly $\alpha_s(k) \leq s \leq \beta_s(k) \leq t$ and $s \leq \beta_t(k) \leq t \leq \alpha_t(k)$ for all $k \geq 0$ (compare also \cref{fig:alpha_beta}).

The following lemma shows that there is a shortest path from~$s$ to~$\alpha_s(k)$ that uses only
vertices in $\{\alpha_s(i), \beta_s(i) \mid 0 \le i \le k\}$.
Here, $N[v] := \set{v, w ; \{v, w\} \in E(G)}$ denotes the closed neighborhood of $v$.

\begin{lemma}\label{thm:alpha_neighbor}
	Let $k \geq 1$. Then, $N[\alpha_s(k)]$ contains $\alpha_s(k-1)$ or $\beta_s(k-1)$.
\end{lemma}
\begin{proof}
  Clearly, the statement trivially holds for~$k=1$ and if~$\alpha_s(k) = \alpha_s(k-1)$.
  Hence, assume that $k \geq 2$ and $\alpha_s(k) \neq \alpha_s(k-1)$.
  In the following, we only consider vertices to the left of $t$, which is why we assume that $t$ is the rightmost vertex for the remainder of this proof.
Let $\{x, \alpha_s(k)\}$ be the last edge of a shortest path from $s$ to~$\alpha_s(k)$.
By definition, we have $\alpha_s(k-1) \leq x$.

If $x < s$, then let $Q$ be a shortest path from $s$ to $\alpha_s(k-1)$.
If $x$ is a vertex of $Q$, then $x = \alpha_s(k-1)$ (otherwise we have a path from~$s$ to~$\alpha_s(k)$ of length at most~$k-1$ implying $\alpha_s(k)=\alpha_s(k-1)$) and we are done.
Otherwise, $Q$ contains an edge $\{y, y'\}$ with $y < x < y'$.
By the \xprop{}, there is an edge between $\alpha_s(k)$ and $y'$,
which yields a path of length at most $k-1$ between $s$ and $\alpha_s(k)$, implying $\alpha_{s}(k-1) = \alpha_s(k)$.

If $x > s$, then let $Q$ be a shortest path from $s$ to $\beta_s(k-1)$.
If~$x$ is a vertex of~$Q$, then $x = \beta_s(k-1)$ (otherwise we have a path from~$s$ to~$\alpha_s(k)$ of length at most~$k-1$ implying $\alpha_s(k)=\alpha_s(k-1)$) and we are done.
Otherwise, $Q$ contains an edge $\{y, y'\}$ with $y < x < y'$.
By the \xprop{}, there is an edge between $\alpha_s(k)$ and $y'$,
which yields a path of length at most $k-1$ between $s$ and $\alpha_s(k)$, implying $\alpha_{s}(k-1) = \alpha_s(k)$.
\end{proof}

\noindent Due to symmetry, \cref{thm:alpha_neighbor} also holds if one exchanges $\alpha$ and $\beta$, and also when $s$ is replaced by $t$.
With the following lemma, we will be able to restrict our search for a shortest $s$-$t$-path mostly to the $\alpha$ and $\beta$ vertices.

\begin{lemma}\label{thm:shortest_path_lemma}
	There is a shortest $s$-$t$-path $P=p_0,\ldots,p_{d^*}$ such that
	\begin{enumerate}[(i)]
		\item $p_i \in \set{\alpha_s(i), \beta_s(i)}$ for all $i < \righti(P)$, and 
		\item $p_i \in \set{\alpha_t(d^*-i), \beta_t(d^*-i)}$ for all $i > \lefti(P)$.
	\end{enumerate}
\end{lemma}

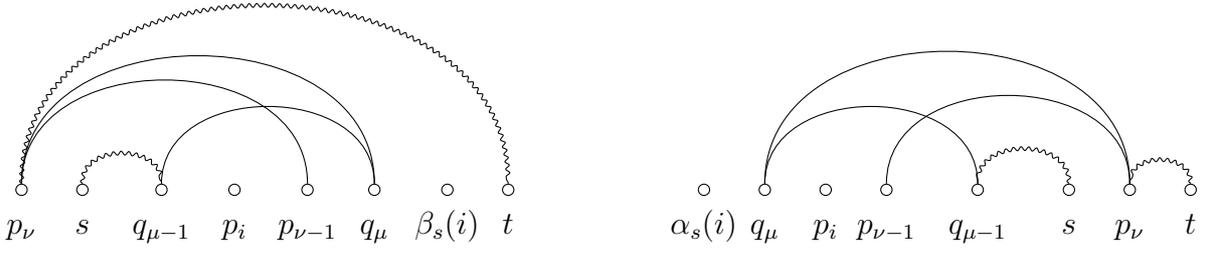
\begin{figure}
  \centering
  \begin{tikzpicture}[xscale=.8]
		\node[vertex,label=below:$\strut{}p_i$] (pi) at (-0.5,0) {};
		\node[vertex,label=below:$\strut{}s$] (s) at (-3,0) {};
		\node[vertex,label=below:$\strut{}t$] (t) at (4,0) {};
		\node[vertex,label=below:$\strut{}\beta_s(i)$] (beta) at (3,0) {};
		\node[vertex,label=below:$\strut{}q_{\mu-1}$] (qmu') at (-1.7,0) {};
		\node[vertex,label=below:$\strut{}p_{\nu-1}$] (pnu') at (0.7,0) {};
		\node[vertex,label=below:$\strut{}p_\nu$] (pnu) at (-4,0) {};
		\node[vertex,label=below:$\strut{}q_\mu$] (qmu) at (1.8,0) {};
		\draw[curveedge] (qmu') to (qmu);
		\draw[curveedge] (pnu) to (pnu');
		\draw[curveedge] (pnu) to (qmu);
		\draw[curvepath] (s) to (qmu');
		\draw[curvepath] (pnu) to (t);
	\end{tikzpicture}\hfill
	\begin{tikzpicture}[xscale=.8]
		\node[vertex,label=below:$\strut{}p_i$] (pi) at (-2,0) {};
		\node[vertex,label=below:$\strut{}s$] (s) at (2,0) {};
		\node[vertex,label=below:$\strut{}t$] (t) at (4,0) {};
		\node[vertex,label=below:$\strut{}\alpha_s(i)$] (a) at (-4,0) {};
		\node[vertex,label=below:$\strut{}q_{\mu-1}$] (qmu') at (0.5,0) {};
		\node[vertex,label=below:$\strut{}p_{\nu-1}$] (pnu') at (-1,0) {};
		\node[vertex,label=below:$\strut{}p_\nu$] (pnu) at (3,0) {};
		\node[vertex,label=below:$\strut{}q_\mu$] (qmu) at (-3,0) {};
		\draw[curveedge] (qmu') to (qmu);
		\draw[curveedge] (pnu) to (pnu');
		\draw[curveedge] (pnu) to (qmu);
		\draw[curvepath] (s) to (qmu');
		\draw[curvepath] (t) to (pnu);
\end{tikzpicture}%
             \caption{Left: The situation described in the proof of \Cref{thm:shortest_path_lemma}. Note that the proof shows that~$\mu=i$ (that is, $q_\mu =\beta_s(i)$) and~$\nu=i+1$ (that is, $p_{\nu-1}=p_i$).
              The zigzag line from~$s$ to~$q_{\mu-1}$ indicates a path of length~$i-1$, whereas the zigzag line from~$p_\nu$ to~$t$ indicates a path of length~$d^*-i-1$. Right: The (similar) situation for the case~$p_i < s$. Note that also here~$\mu=i$ and~$\nu=i+1$ holds (that is, $p_{\nu-1}=p_i$ and~$q_\mu=\alpha_s(i)$).}
\label{fig:prop(i)}
\end{figure}

\begin{proof}
	We begin by showing that there always exists a shortest $s$-$t$-path satisfying~(i) (see \Cref{fig:prop(i)} for an illustration).
	Assume towards a contradiction, that no shortest $s$-$t$-path satisfies~(i).
	Then, let $P=p_0,p_1,\dots,p_{d^*}$ be a shortest $s$-$t$-path chosen such that it maximizes $i$, where $i < \righti(P)$ is the smallest index such that $p_i\not\in\set{\alpha_s(i), \beta_s(i)}$.
	Since $i < \righti(P)$, we know that all of $p_0, \dots, p_i$ are to the left of $t$ by \cref{cor:vertex-st-order}.
        We assume that $s < p_i$ (the case $s > p_i$ is similar and uses that~$p_i\neq \alpha_s(i)$).
	Since $p_i \neq \beta_s(i)$, it holds $p_i  < \beta_s(i)$.
        Note that also~$\beta_s(i) < t$, since $\beta_s(i) = t$ would contradict $\dist(s, \beta_s(i)) = i < d^* = \dist(s, t)$.
	By \cref{thm:alpha_neighbor}, there exists a shortest path $Q=q_0,q_1,\dots,q_i$ from $s=q_0$ to $\beta_s(i)=q_i$ where~$q_j\in\{\alpha_s(j),\beta_s(j)\}$ for all $0\le j \le i$.
	Let $\mu \le i$ be an index such that $q_{\mu-1} < p_i < q_\mu$.
	Further, let $\nu > i$ be the minimal index for which $p_\nu \notin [q_{\mu-1}, q_\mu]$ (that is, $p_{\nu-1}\in [q_{\mu-1}, q_\mu]$).
	Then, $\{q_{\mu-1}, q_\mu\}$ and $\{p_{\nu-1}, p_\nu\}$ are crossing edges.
	By the \xprop{}, $G$ contains the edge $\{q_{\mu'}, p_\nu\}$, where either $\mu' = \mu$ (if $p_\nu < q_{\mu-1} < p_{\nu-1} < q_\mu$) or  $\mu' = \mu-1$ (if $q_{\mu-1} < p_{\nu-1} < q_\mu < p_\nu$).
	Hence, we obtain an $s$-$t$-walk $P'=q_0,\ldots,q_{\mu'},p_\nu,\ldots,p_{d^*}$ of length $\mu' + 1 + (d^*-\nu)$.
	Clearly, this length must be at least~$d^*$, which implies $\mu'+1 \ge \nu$.
        Since~$\mu' \le i$ and~$\nu \ge i+1$, this is only possible for $\mu' = \mu = i$ and $\nu = i + 1$.
        Hence, $P'$ has length~$d^*$, which means that~$P'$ is a shortest path from~$s$ to~$t$ (with $i < \righti(P')$).
	Recall that, by construction, $q_j\in\{\alpha_s(j),\beta_s(j)\}$ holds for all $0\le j \le i$, which contradicts the choice of $P$.
	This proves the existence of a shortest $s$-$t$-path satisfying~(i).

        Before moving on, we show that if the chosen shortest path~$P$ in the above argument satisfies~(ii), then also~$P'$ does so. This observation will be helpful for the second part of the proof.        
	To this end, note that we concluded above that $\mu' = \mu = i$ and $\nu = i+1$, which implies that $p_{i+1} < p_i$.
	Assume towards a contradiction that $p_j < p_{i+1}$ for some~$j < i$.
	Then, there must be a pair of crossing edges from the subpaths from $p_j$ to $p_i$ and from $p_{i+1}$ to $t$.
	By the \xprop{}, there is now an edge between two non-consecutive vertices of $P$, contradicting the fact that $P$ is an induced path.
	Thus, $p_{i+1}$ must be to the left of $p_0, \dots, p_i$.
        This implies $i < \lefti(P)$, and thus, $p_{\lefti(P)}$ is also a vertex of $P'$.
	Furthermore, all of $q_0, \dots, q_i$ must be to the right of $p_{\lefti(P)}$ because otherwise $P'$ would contain a pair of crossing edges from the subpaths from $s$ to $q_i$ and from $p_{\lefti(P)}$ to $t$. Thus, $\lefti(P) = \lefti(P')$.
	We conclude that $i < \lefti(P) = \lefti(P')$, that is, $P'$ also satisfies~(ii).
	
	Now, assume towards a contradiction that no shortest $s$-$t$-path satisfies both (i) and~(ii).
	Then, let~$P=p_0,p_1,\dots,p_{d^*}$ be a shortest $s$-$t$-path that satisfies~(i) and which is chosen to minimize~$i > \lefti(P)$, which is the largest index such that $p_i \not\in \set{\alpha_t(d^*-i), \beta_t(d^*-i)}$.
	Now, if we reverse the vertex ordering of~$G$, swap $s$ and $t$, and reverse the direction of $P$, then we obtain a shortest $s$-$t$-path $\tilde{P}$ satisfying~(ii) but not~(i).
        Applying the above arguments to~$\tilde{P}$ now yields a contradiction.
        
\end{proof}

Note that the only case in which \cref{thm:shortest_path_lemma} does not restrict all vertices to be $\alpha$ or~$\beta$ vertices is when $\righti(P) < \lefti(P)$, in which case \cref{thm:right_left} applies.
This will allow us to show that, when searching for a shortest $s$-$t$-path, it is sufficient to check for three possible cases (illustrated in \cref{fig:path_cases}).
In the following, we define the \textit{left} and \textit{right horizon} of a vertex $v$ as the furthest neighbor in that direction, that is, $\lhorizon{v} := \min(N[v])$ and $\rhorizon{v} := \max(N[v])$.
Remember that $d^* = \dist(s, t)$.

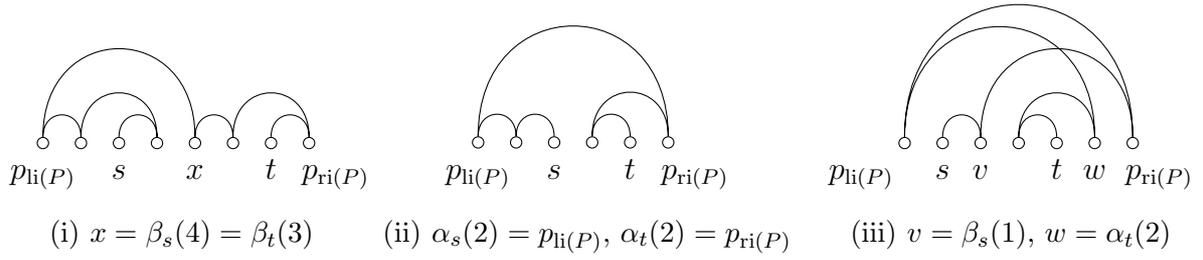
\begin{figure}
\centering
\begin{subfigure}[b]{0.3\textwidth}
	\centering
	\begin{tikzpicture}[xscale=0.5, label distance=-3]
		\node[vertex,label=below:$\strut{}p_{\lefti(P)}$] (0) at (0,0) {};
		\node[vertex] (1) at (1,0) {};
		\node[vertex,label=below:$\strut{}s$] (2) at (2,0) {};
		\node[vertex] (3) at (3,0) {};
		\node[vertex,label=below:$\strut{}x$] (4) at (4,0) {};
		\node[vertex] (5) at (5,0) {};
		\node[vertex,label=below:$\strut{}t$] (6) at (6,0) {};
		\node[vertex,label={[xshift=10]below:$\strut{}p_{\righti(P)}$}] (7) at (7,0) {};
		\draw[curveedge] (2) to (3) to (1) to (0) to (4) to (5) to (7) to (6);
	\end{tikzpicture}%
	\caption{$x = \beta_s(4) = \beta_t(3)$}
\end{subfigure}%
\begin{subfigure}[b]{0.37\textwidth}
	\centering
	\begin{tikzpicture}[xscale=0.5, label distance=-3]
		\node[vertex,label=below:\strut{}$p_{\lefti(P)}$] (0) at (0,0) {};
		\node[vertex] (1) at (1,0) {};
		\node[vertex,label=below:$\strut{}s$] (2) at (2,0) {};
		\node[vertex] (3) at (3,0) {};
		\node[vertex,label=below:$\strut{}t$] (4) at (4,0) {};
		\node[vertex,label={[xshift=10]below:$\strut{}p_{\righti(P)}$}] (5) at (5,0) {};
		\draw[curveedge] (2) to (1) to (0) to (5) to (3) to (4);
	\end{tikzpicture}%
	\caption{$\alpha_s(2) = p_{\lefti(P)}$, $\alpha_t(2) = p_{\righti(P)}$}
\end{subfigure}%
\begin{subfigure}[b]{0.33\textwidth}
	\centering
	\begin{tikzpicture}[xscale=0.5, label distance=-3]
		\node[vertex,label=-95:$\strut{}p_{\lefti(P)}$] (0) at (0,0) {};
		\node[vertex,label=below:$\strut{}s$] (1) at (1,0) {};
		\node[vertex,label=below:$\strut{}v$] (2) at (2,0) {};
		\node[vertex] (3) at (3,0) {};
		\node[vertex,label=below:$\strut{}t$] (4) at (4,0) {};
		\node[vertex,label=below:$\strut{}w$] (5) at (5,0) {};
		\node[vertex,label={[xshift=10]below:$\strut{}p_{\righti(P)}$}] (6) at (6,0) {};
		\draw[curveedge] (1) to (2) to (6) to (0) to (5) to (3) to (4);
	\end{tikzpicture}%
	\caption{$v = \beta_s(1)$, $w = \alpha_t(2)$}
\end{subfigure}
\caption{The three cases distinguished by \cref{thm:shortest_path_conditions_sufficient} (only the vertices of a shortest path are drawn).}
\label{fig:path_cases}
\end{figure}

\begin{lemma}\label{thm:shortest_path_conditions_sufficient}
	Each of the following conditions implies that $d^* \leq d$.
	\begin{enumerate}[(i)]
		\item $\beta_s(\sigma) = \beta_t(\tau)$ with $\sigma + \tau = d$.
		\item $G$ contains an edge between $\alpha_s(\sigma)$ and $\alpha_t(\tau)$ with $\sigma + \tau = d-1$.
		\item There exist $v \in \set{\alpha_s(\sigma), \beta_s(\sigma)}$ and $w \in \set{\alpha_t(\tau), \beta_t(\tau)}$ with $v < w$ and $\sigma + \tau = d-3$ such that $\rhorizon{v} \geq t$ and $\lhorizon{w} \leq s$.
	\end{enumerate}
\end{lemma}
\begin{proof}
	For (i) and (ii), the existence of an $s$-$t$-path of length at most $d$ follows directly from the definition of $\alpha$ and $\beta$.
	Assume now that (iii) is true.
	We may also assume $s < w$ since $s = w$ would directly imply $d^* = \tau < d$.
	Define $h_v := \rhorizon{v}$ and $h_w := \lhorizon{w}$.
	If the edges $\{v, h_v\}$ and $\{w, h_w\}$ are crossing, then the existence of an $s$-$t$-path of length $\sigma + \tau + 3$ follows from the \xprop{}.
        Otherwise, let $v \leq h_w \leq s < w$ (the case $v < t \le h_v \le w$ is symmetric).
	Let $P$ be a shortest $s$-$v$-path. Then, $P$ either crosses the edge $\{w, h_w\}$ or contains~$w$ or~$h_w$.
	In any case, we obtain an $s$-$w$-path of length at most $\sigma + 1$ and thus an $s$-$t$-path of length at most $\sigma + \tau + 1 < d$.
\end{proof}

\begin{lemma}\label{thm:shortest_path_conditions_necessary}
	At least one of the conditions (i)--(iii) in \cref{thm:shortest_path_conditions_sufficient} holds for $d = d^*$.
\end{lemma}
\begin{proof}
  Let $P=p_0,p_1,\dots,p_{d^*}$ be a shortest $s$-$t$-path satisfying conditions~(i) and~(ii) as described in \cref{thm:shortest_path_lemma}.
  
  If $\lefti(P) < \righti(P) - 1$, then, for $\sigma := \lefti(P) + 1$, it holds $0 < \sigma < d^*$ and $p_\sigma \in \set{\alpha_s(\sigma), \beta_s(\sigma)} \cap \set{\alpha_t(d^*-\sigma), \beta_t(d^*-\sigma)}$.
  By definition, for all~$j$, we have
  \begin{align*}
    &\alpha_s(j) \leq s < t \leq \alpha_t(j),\\
    &s \le \beta_s(j) \le t,\\
    &s \le \beta_t(j) \le t.
  \end{align*}
  Clearly, $p_\sigma \notin \{s, t\}$, thus we must have $p_\sigma = \beta_s(\sigma) = \beta_t(d^* - \sigma)$, and therefore (i) is true with $\tau := d^* - \sigma$.

  If $\lefti(P) = \righti(P) -1$, then~$P$ contains the edge~$\{p_{\lefti(P)},p_{\righti(P)}\}$.
  Moreover, $p_{\lefti(P)} = \alpha_s(\lefti(P))$ and $p_{\righti(P)} = \alpha_t(d^*-\righti(P))$ clearly holds.
  Hence, (ii) is true with $\sigma := \lefti(P)$ and $\tau := d^* - \righti(P)$.
  
  If $\righti(P) < \lefti(P)$, then, by \cref{thm:right_left}, the two edges $\{p_{\righti(P)-1},p_{\righti(P)}\}$ and $\{p_{\lefti(P)},p_{\lefti(P)+1}\}$ are crossing, that is, $p_{\righti(P)-1} < p_{\lefti(P)+1}$ and $\righti(P) = \lefti(P) -1$.
  Thus, (iii) holds with $v := p_{\righti(P)-1}$, $w := p_{\lefti(P)+1}$,
  $\sigma := \righti(P)-1$, and $\tau := d^* - (\lefti(P) + 1)$.
  Note that $\rhorizon{p_{\righti(P)-1}} \geq p_{\righti(P)} \geq t$ and $\lhorizon{p_{\lefti(P)+1}} \leq p_{\lefti(P)} \leq s$.
\end{proof}

We are now ready to present our shortest path algorithm (\cref{alg:shortest_path}).
It uses the following two definitions of the neighbors of a vertex~$v$ that are closest to another vertex~$a$.
\begin{align*}
\rclosest{a}{v} &:= \min\set{x \in N[v] \mid x \geq a}, \\
\lclosest{a}{v} &:= \max\set{x \in N[v] \mid x \leq a}.
\end{align*}
If the sets on the right hand side are empty, we treat these values as $\infty$ resp.\ $-\infty$.
For given vertices $v$ and $a$, these two neighbors can be computed in $\bigO(\log(\deg(v)))$ time using binary search on the (sorted) adjacency list of $v$.
It is also possible to compute them for a fixed vertex~$v$ and all other vertices $a$ in $\bigO(n)$ time.

\cref{alg:shortest_path} iteratively computes the values $\alpha_s(k)$, $\beta_s(k)$, $\alpha_t(k)$, and $\beta_t(k)$ for increasing values of $k$ and checks each time, whether one of the three cases of \cref{thm:shortest_path_conditions_sufficient} holds.
Condition~(i) is easy to test. To test condition (iii), the algorithm stores the values
\begin{align*}
  q_s &:= \min\set{\sigma \le k ; \rhorizon{\alpha_s(\sigma)} \geq t \;\lor\; \rhorizon{\beta_s(\sigma)} \geq t} \text{ and}\\
  q_t &:= \min\set{\tau \le k; \lhorizon{\alpha_t(\tau)} \leq s \;\lor\; \lhorizon{\beta_t(\tau)} \leq s}.
\end{align*}
In contrast, checking condition (ii) requires slightly more effort.
To test for the existence of an edge~$\{\alpha_s(\sigma),\alpha_t(\tau)\}$ with~$\sigma \le \tau$, we store the computed values for $\alpha_s(\sigma)$ in a list~$Q_s$, sorted in such a way that the vertices $r_\sigma := \rclosest{t}{\alpha_s(\sigma)}$ are ordered from left to right.
Whenever the iteration reaches a value $\tau$ such that $\alpha_t(\tau)$ is to the right of the leftmost of these $r_\sigma$, we have an edge of the form $\{\alpha_s(\sigma), \alpha_t(\tau)\}$ with $\sigma \leq \tau$ (by the \xprop{}).
To test for edges $\{\alpha_s(\sigma),\alpha_t(\tau)\}$ with $\sigma \geq \tau$, we use a symmetrical procedure with a list~$Q_t$.

\begin{algorithm}[t]
\begin{algorithmic}[1]
\Input{A terrain-like graph~$G$ and two vertices $s < t$.}
\Output{$\dist(s, t)$.}
\State $\alpha_s \gets \beta_s \gets s$
\Comment{$\alpha_s(k)$, $\beta_s(k)$ for the current value of $k$}
\State $\alpha_t \gets \beta_t \gets t$
\Comment{$\alpha_t(k)$, $\beta_t(k)$ for the current value of $k$}
\State $d_s[v] \gets d_t[v] \gets \infty \ \forall v \in V(G)$\label{line:init}
\Comment{Distances from $s$ resp.\ $t$}
\State $d \gets \infty$
\Comment{Smallest $s$-$t$-distance found so far}
\State $Q_s \gets [s]$
\Comment{Candidate vertices for test (ii) ($\sigma \leq \tau$), sorted by ascending $\rclosest{t}{}$}
\State $Q_t \gets [t]$
\Comment{Candidate vertices for test (ii) ($\sigma \geq \tau$), sorted by descending $\lclosest{s}{}$}
\State $q_s \gets q_t \gets \infty$
\Comment{Distances for test (iii)}

\For{$k \gets 0, 1, \dots, d$}\label{line:for}
	\State \Call{UpdateDistances}{$k$}
	\Comment{Updates $d_s$, $d_t$, $Q_s$, $Q_t$, $q_s$, and $q_t$}
	\If{$\beta_s \geq \beta_t$}
		\Comment{Test (i)}
		\State $d \gets \min\set{d, \; d_s[\beta_s] + d_t[\beta_s], \; d_s[\beta_t] + d_t[\beta_t]}$ \label{line:update_i}
	\EndIf
	\While{$\alpha_t \geq \rclosest{t}{\head(Q_s)}$}
		\Comment{Test (ii) for $\sigma \le \tau$}
		\State $d \gets \min\set{d, \; d_s[\head(Q_s)] + d_t[\alpha_t] + 1}$ \label{line:update_iia}
		\State remove the head from $Q_s$
	\EndWhile
	\While{$\alpha_s \leq \lclosest{s}{\head(Q_t)}$}
		\Comment{Test (ii) for $\sigma \ge \tau$}
		\State $d \gets \min\set{d, \; d_t[\head(Q_t)] + d_s[\alpha_s] + 1}$ \label{line:update_iib}
		\State remove the head from $Q_t$
	\EndWhile
	\If{$q_s < \infty$ and $q_t < \infty$}
		\Comment{Test (iii)}
		\State $d \gets \min\set{d, \; q_s + q_t + 3}$ \label{line:update_iii}
	\EndIf
	\State \Call{ExtendSearchRange}{}
	\Comment{Updates $\alpha_s$, $\beta_s$, $\alpha_t$, and $\beta_t$ to $k+1$}
\EndFor
\State \Return $d$
\end{algorithmic}
\caption{Shortest Path}
\label{alg:shortest_path}
\end{algorithm}

\begin{algorithm}[t]
\begin{algorithmic}[1]
\Procedure{UpdateDistances}{$k$}
	\State $d_s[\alpha_s] \gets \min\set{k, d_s[\alpha_s]}$
	\State $d_t[\alpha_t] \gets \min\set{k, d_t[\alpha_t]}$
	\State $d_s[\beta_s] \gets \min\set{k, d_s[\beta_s]}$
	\State $d_t[\beta_t] \gets \min\set{k, d_t[\beta_t]}$
	\If{$\max\set{\rhorizon{\alpha_s}, \; \rhorizon{\beta_s}} \geq t$}
		\State $q_s \gets \min\{q_s, k\}$
	\EndIf
	\If{$\min\set{\lhorizon{\alpha_t}, \; \lhorizon{\beta_t}} \leq s$}
		\State $q_t \gets \min\{q_t, k\}$
	\EndIf
	\If{$\rclosest{t}{\alpha_s} < \rclosest{t}{\head(Q_s)}$}
		\State add $\alpha_s$ to the front of $Q_s$
	\EndIf
	\If{$\lclosest{s}{\alpha_t} > \lclosest{s}{\head(Q_t)}$}
		\State add $\alpha_t$ to the front of $Q_t$
	\EndIf
\EndProcedure
\end{algorithmic}
\label{alg:shortest_path_updatedistances}
\caption{\textsc{UpdateDistances}}
\end{algorithm}

\begin{algorithm}[t]
\begin{algorithmic}[1]
\Procedure{ExtendSearchRange}{}
	\State $\alpha_s' \gets \min_{x \in \set{\alpha_s, \beta_s}}(\lhorizon{x})$
	\State $\beta_s' \gets \max_{x \in \set{\alpha_s, \beta_s}}(\lclosest{t}{x})$
	\State $\alpha_s \gets \alpha_s'$
	\State $\beta_s \gets \beta_s'$
	\State $\alpha_t' \gets \max_{x \in \set{\alpha_t, \beta_t}}(\rhorizon{x})$
	\State $\beta_t' \gets \min_{x \in \set{\alpha_t, \beta_t}}(\rclosest{s}{x})$
	\State $\alpha_t \gets \alpha_t'$
	\State $\beta_t \gets \beta_t'$
\EndProcedure
\end{algorithmic}
\label{alg:shortest_path_extendsearchrange}
\caption{\textsc{ExtendSearchRange}}
\end{algorithm}

\begin{theorem}
	\Cref{alg:shortest_path} computes $d^* = \dist(s, t)$ in $\bigO(d^*)$ time (excluding the time to compute all $\lclosest{a}{v}$ and $\rclosest{a}{v}$).
\end{theorem}
\begin{proof}
	We begin by showing the following invariants that are maintained throughout the computation.
	\begin{enumerate}
		\item At the beginning of each iteration of the for loop (\cref{line:for}), we have $\alpha_s = \alpha_s(k)$, $\alpha_t = \alpha_t(k)$, $\beta_s = \beta_s(k)$, and $\beta_t = \beta_t(k)$.
			This is ensured by the \textsc{ExtendSearchRange} procedure (based on \cref{thm:alpha_neighbor}).
		\item If $d_s[v] < \infty$ holds at some point, then $d_s[v] = \dist(s, v)$ (the same holds for~$d_t[v]$).
			This is ensured by the \textsc{UpdateDistances} procedure.
		\item At every point, $d \geq d^*$ holds.
			This is obvious for \cref{line:update_i}.
		
			For \cref{line:update_iia}, let $ \alpha_s(h) = \head(Q_s)$, $r_h = \rclosest{t}{\alpha_s(h)}$, and 
			let $P$ be a shortest path from $t$ to $\alpha_t(k)$.
			Then, either $P$ contains $r_h$ or crosses the edge $\{\alpha_s(h), r_h\}$ in which case the \xprop{} implies the existence of an edge between~$\alpha_s(h)$ and another vertex on~$P$.
			In any case $\dist(t, \alpha_s(h)) \leq \dist(t, \alpha_t(k)) + 1$.
                        Thus, $\dist(s,t) \le \dist(s,\alpha_s(h)) + \dist(t, \alpha_t(k)) + 1$.
			A symmetrical argument works for \cref{line:update_iib}.
			
			For \cref{line:update_iii}, if $q_s + q_t < \infty$, then, according to \textsc{UpdateDistances}, we found vertices $v_s\in\{\alpha_s(q_s),\beta_s(q_s)\}$ and $v_t\in\{\alpha_t(q_t),\beta_t(q_t)\}$ 
			with $\dist(s, v_s) = q_s$ and $h_s := \rhorizon{v_s} \geq t$,
			and $\dist(t, v_t) = q_t$ and $h_t := \lhorizon{v_t} \leq s$.
                        We can assume that $\beta_s(k) < \beta_t(k)$, since otherwise test (i) would already have succeeded and set $d \leq q_s + q_t$.
                        Therefore, we have $v_s \leq \beta_s(k) < \beta_t(k) \leq v_t $.
                        Thus, condition~(iii) of \Cref{thm:shortest_path_conditions_sufficient}
                        holds which implies $\dist(s,t) \leq q_s + q_t + 3$.
                        
                      \end{enumerate}

        Having established the above invariants, we claim that $d \leq d^*$ holds after the for loop has reached $k=d^*$.
        By \cref{thm:shortest_path_conditions_necessary}, we know that one of the conditions (i)--(iii) is true.

        If condition (i) is true, then there are $\sigma$ and $\tau$ with $\sigma + \tau = d^*$ and $\beta_s(\sigma) = \beta_t(\tau)$.
	Then, $d$ will be set to $d^*$ in \cref{line:update_i} as soon as $k$ reaches $\max\{\sigma, \tau\}$.
	
	If condition (ii) is true, then there are $\sigma$ and $\tau$ with $\sigma + \tau = d^*-1$ such that $\alpha_s(\sigma)$ and $\alpha_t(\tau)$ are connected by an edge.
	As soon as $k = \sigma$, the vertex $\alpha_s(\sigma)$ will be added to~$Q_s$ by \textsc{UpdateDistances} and as soon as $k = \tau$, the vertex $\alpha_t(\tau)$ will be added to $Q_t$.
	Therefore, when $k$ reaches the value $\max\{\sigma, \tau\}$, then $d$ will be set to~$d^*$ either in \cref{line:update_iia} or \cref{line:update_iib} (depending on whether $\sigma \leq \tau$).
	
	If condition (iii) is true, then there are $\sigma$ and $\tau$ with $\sigma + \tau = d^*-3$ 
	and $v \in \set{\alpha_s(\sigma), \beta_s(\sigma)}$, $w \in \set{\alpha_t(\tau), \beta_t(\tau)}$ such that $\rhorizon{v} \geq t$ and $\lhorizon{w} \leq s$.
	Then, $q_s$ will be set to $\sigma$ by \textsc{UpdateDistances} when $k=\sigma$ and $q_t$ will be set to $\tau$ when $k = \tau$.
	Therefore, when $k$ reaches the value $\max\{\sigma, \tau\}$, then $d$ will be set to $q_s + q_t + 3  = d^*$ in \cref{line:update_iii}.
	
	Since $d \geq d^*$ is an invariant, the for loop will at some point reach $k = d^*$, and thus, \Cref{alg:shortest_path} outputs exactly~$d^*$.
	
	It remains to prove the running time.
	As shown above, the for loop is repeated at most $d^*+1$ times.
	Apart from the inner while loops, each iteration only takes constant time.
	Since at most one vertex is added to $Q_s$ and $Q_t$ in each of these iterations, the inner while loops are also iterated at most $d^*+1$ times overall.

	Note that initializing all $d_s[v]$ and $d_t[v]$ with $\infty$ in \cref{line:init} is only done for ease of notation.
	In practice, these distance values only need to be stored once they get updated by \textsc{UpdateDistances}.
\end{proof}

How does the running time change if we account for the computation of $\lclosest{a}{v}$ and $\rclosest{a}{v}$?
\Cref{alg:shortest_path} requires knowledge of these values for all $a \in \{s, t\}$ and $v \in \set{\alpha_b(i), \beta_b(i) ; b \in \{s, t\}, i \leq d^*}$.
This amounts to $\bigO(d^*)$ many computations, each taking $\bigO(\log(\Delta))$ time where $\Delta$ is the maximal degree of any vertex.
The overall running time is thus in $\bigO(d^* \cdot \log(\Delta))$.
Alternatively, one can precompute these values for all $a, v \in V(G)$ in $\bigO(n^2)$ time and then run \cref{alg:shortest_path} in output-optimal time $\bigO(d^*)$ for any pair $s, t$.

Note that even though we always assumed $\dist(s, t) < \infty$, \cref{alg:shortest_path} can easily be extended to handle the case $\dist(s, t) = \infty$ by additionally testing whether any of the vertices $\alpha_s, \beta_s, \alpha_t, \beta_t$ was updated in \textsc{ExtendSearchRange}. If not, then the algorithm terminates.

\section{Dominating Set on Funnel Visibility Graphs}\label{sec:dominatingset}

In this section we consider the \textsc{Dominating Set} problem (which is a variant of the \textsc{Art Gallery} or \textsc{Guarding} problem in the context of visibility graphs) on a subclass of terrain visibility graphs called \emph{funnel} (or \emph{tower} \cite{ColleyTowers97}) visibility graphs~\cite{ChoiFunnels}.

\problemdef{Dominating Set}
{An undirected graph~$G$ and an integer~$k$.}
{Does~$G$ contain~$k$ vertices such that every other vertex is adjacent to at least one of them?}

A funnel is a terrain that has exactly one convex vertex (called the \emph{bottom}) and whose leftmost and rightmost vertex see each other\footnote{We remark that our results are easily extended to the case when the two outer vertices do not see each other.} (see \Cref{fig:funnel}).
Funnels appear in several visibility-related tasks in geometry and their visibility graphs are linear-time recognizable~\cite{ChoiFunnels}. Funnel visibility graphs are characterized precisely as bipartite permutation graphs with an added Hamiltonian cycle~\cite{ColleyTowers97}.
In the following, we assume to be given the graph together with the corresponding vertex coordinates of the funnel (which can be precomputed in linear time~\cite{ChoiFunnels}).

\textsc{Dominating Set} is NP-hard on polygon visibility graphs~\cite{LS95} but it is open whether NP-hardness also holds on terrain visibility graphs.
\textsc{Dominating Set} is solvable in linear time on permutation graphs~\cite{OptPermDom}, but due to the added Hamilton cycle, that approach seems not to be applicable for funnel visibility graphs.
We give an~$\bigO(n^4)$-time algorithm solving \textsc{Dominating Set} on funnel visibility graphs.

Let $L=\{\lambda_0,\ldots,\lambda_{n_L}\}$ and $R=\{\rho_0,\ldots,\rho_{n_R}\}$ be the vertices to the left resp.\ right of the bottom vertex $\lambda_0=\rho_0$, where~$L$ and~$R$ are ordered by increasing $y$-coordinate.
$L$ and $R$ are also referred to as the two \emph{chains} of the funnel.
We define $\ind(\lambda_j) = \ind(\rho_j) = j$ as the index of a vertex in its corresponding chain.
The following observation is immediate.

\begin{observation}\label{obs:domination_obs}
  For each vertex $v$, the sets $N[v] \cap L$ and $N[v] \cap R$ are both consecutive subsets of $L$ respectively $R$, that is,
    $N[v]\cap L = \{\lambda_i| i\in[l_L,u_L]\}$ and $N[v]\cap R = \{\rho_i| i\in[l_R,u_R]\}$ for some $0 \le l_L \le u_L \le n_L$ and $0 \le l_R \le u_R \le n_R$.
\end{observation}

\begin{figure}
  \centering
\begin{tikzpicture}
  \tikzset{every node/.style={vertex}}
  \node[label=below:$\lambda_0$] (l0) at (0,0) {};
  \node[label=below:$\lambda_1$] (l1) at (-1,2) {};
  \node[label=below:$\lambda_2$] (l2) at (-2,3) {};
  \node[label=below:$\lambda_3$] (l3) at (-3,3.5) {};
  \node[label=right:$\rho_1$] (r1) at (0.5,1.7) {};
  \node[label=-45:$\rho_2$] (r2) at (1,2.5) {};
  \node[label=right:$\rho_3$] (r3) at (2,2.8) {};
  \draw[thick] (l3) -- (l2) -- (l1) -- (l0) -- (r1) -- (r2) -- (r3);
  \draw (l1) -- (r1);
  \draw (l1) -- (r2);
  \draw (l2) -- (r1);
  \draw (l2) -- (r2);
  \draw (l2) -- (r3);
  \draw (l3) -- (r2);
  \draw (l3) -- (r3);
\end{tikzpicture}
\caption{A funnel visibility graph.}
\label{fig:funnel}
\end{figure}
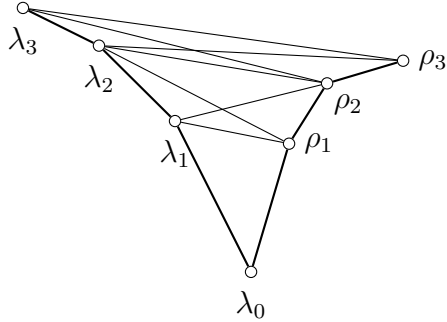

\newcommand{\ldown}[1]{\prescript{}{\downarrow}{#1}}
\newcommand{\rdown}[1]{{#1}_{\downarrow}}
\newcommand{\lup}[1]{\prescript{}{\uparrow}{#1}}
\newcommand{\rup}[1]{{#1}_{\uparrow}}

For~$\ell\in [n_L]$, $r\in[n_R]$, we define $L_{\ell} := \{\lambda_0,\ldots,\lambda_\ell\}$ and $R_{r} := \{\rho_0,\ldots,\rho_r\}$.
We call a vertex subset $U \subseteq (L\cup R)$ a \emph{base} if $U = L_\ell \cup R_r$ for some $\ell, r$.
For a vertex~$v$, define $\lup{v}\in N[v] \cap L$ (resp.\ $\rup{v}\in N[v] \cap R$) as the neighbor with the highest index on~$L$ (resp.\ $R$)
and $\ldown{v}\in N[v] \cap L$ (resp.\ $\rdown{v}\in N[v] \cap R$) as the neighbor on~$L$ (resp.\ $R$) with the lowest index.

We compute a minimum dominating set via dynamic programming. To this end, we first prove the following structural lemma.

\begin{lemma}\label{lem:domination_base}
	Let $W := L_\ell \cup R_r$ be a base and let $d(\ell) \in N[\lambda_\ell]$ and $d(r) \in N[\rho_r]$ be two vertices.
	Then, at least one of the subsets $W \setminus N[d(\ell)]$, $W \setminus N[d(r)]$, or
        $W \setminus (N[d(\ell)] \cup N[d(r)])$ is a base.
\end{lemma}
\begin{proof}
If $\ind(\rup{d(\ell)}) \geq r$, then $W\setminus N[d(\ell)]$ is clearly a base, since by \Cref{obs:domination_obs} $N[d(\ell)]\cap L_\ell$ and $N[d(\ell)]\cap R_r$ are sets of consecutive vertices.
Similarly, if $\ind(\lup{d(r)}) \geq \ell$, then $W\setminus N[d(r)]$ is a base.

Otherwise (that is, $\ind(\rup{d(\ell)}) < r$ and $\ind(\lup{d(r)}) < \ell$), if $d(\ell)$ and $d(r)$ are both from the same chain (say $L$),
then $\ind(d(\ell)) > \ind(d(r))$ (since otherwise $\ind(\lup{d(r)}) \geq \ell$) and therefore $\ind(\rup{d(\ell)}) \geq \ind(\rup{d(r)}) \geq r$, which is a contradiction.

Hence, assume $d(\ell)$ and $d(r)$ are from different chains.
If $d(\ell) \in L$ and $d(r) \in R$, then we must have $\ind(\lup{d(r)}) \leq \ell - 1 \leq \ind(d(\ell))$ and $\ind(\rup{d(\ell)}) \leq r - 1 \leq \ind(d(r))$.
Then, $d(\ell)$ and~$d(r)$ are connected by an edge (by the \xprop{}).
If $d(r)\in L$ and $d(\ell)\in R$, then $\ind(\rup{d(r)}) \geq r > \ind(\rup{d(\ell)}) \geq \ind(d(\ell))$ and $\ind(\lup{d(\ell)}) \geq \ell > \ind(\lup{d(r)}) \geq \ind(d(r))$, 
so $d(\ell)$ and $d(r)$ share an edge (by the \xprop{}).
It follows that $W \setminus (N[d(\ell)] \cup N[d(r)])$ is a base (using \Cref{obs:domination_obs}).
\end{proof}

\begin{theorem}
  \textsc{Dominating Set} is solvable in~$\bigO(n^4)$ time on funnel visibility graphs.
\end{theorem}

\begin{proof}
We define $D(\ell, r)$ to be a minimum-size subset of~$L\cup R$ that dominates $L_\ell \cup R_r$.
Clearly, $D(\ell,r)$ has to contain some $x \in N[\lambda_\ell]$ and~$y \in N[\rho_r]$.
We try out all possible choices for~$x$ and~$y$. For a fixed choice, there are three cases for possible candidate sets:
If $\ind(\rup{x}) \geq r$, then the candidate set is
$$\{x\} \cup D(\ind(\ldown{x})-1, \min\{r, \ind(\rdown{x})-1\}).$$
Otherwise, if $\ind(\lup{y}) \geq \ell$, then the candidate set is
$$\{y\} \cup D(\min\{\ell, \ind(\ldown{y})-1\}, \ind(\rdown{y})-1).$$
Finally, if $\ind(\rup{x}) < r$ and $\ind(\lup{y}) < \ell$, then the candidate set is
$$\{x, y\} \cup D(\min\{\ind(\ldown{x}), \ind(\ldown{y})\}-1).$$
Note that \Cref{lem:domination_base} guarantees that the above three cases are well-defined, that is, we always consider dominating sets of bases in the recursion.
The recursion terminates with $D(-1,-1):=\emptyset$.
To compute~$D(\ell,r)$, we keep the minimum-size candidate set over all possible choices for~$x$ and $y$.
Hence, a minimum dominating set~$D(n_L,n_R)$ can be computed in~$\bigO(n^2\cdot n^2)$ time.
\end{proof}

\section{Conclusion}

Several open questions remain.
Most prominently, a precise characterization of terrain visibility graphs (and their polynomial-time recognition) still remains open.
It might also be interesting to give a characterization of induced subgraphs of terrain visibility graphs. Note that these clearly still satisfy the X-property for some vertex ordering, but not necessarily the bar-property.
For example, it is open whether all unit interval graphs can appear as induced graphs.

As regards algorithmic questions, polynomial-time solvability of \textsc{Dominating Set}
on terrain visibility graphs is open.
One might also improve the running time for \textsc{Dominating Set} on funnel visibility graphs.
Furthermore, a fast algorithm for finding shortest paths with respect to Euclidean distances on arbitrary induced subgraphs of terrain visibility graphs might be of interest. 
In general, it is worth to search for more efficient algorithms to compute graph characteristics used in practice such as clustering coefficients or centrality measures.

{
\raggedright
\printbibliography
}

\end{document}